\documentclass[twocolumn, pra]{revtex4}
\usepackage{amsmath,amssymb,amsthm,epsfig,graphics,graphicx,hyperref}
\usepackage{verbatim}
\newtheorem{theorem}{Theorem}
\newtheorem{lemma}{Lemma}

\theoremstyle{definition}

\newcommand{\ket}[1]{|#1\rangle}

\begin{document}

\title{Randomly distilling W-class states into general configurations of two-party entanglement}

\author{W. Cui}
\email{cuiwei@physics.utoronto.ca}
\author{E. Chitambar}
\email{e.chitambar@utoronto.ca}
\author{H.K. Lo}
\email{hklo@comm.utoronto.ca}
\thanks{\\ $^*{}^\dagger$ Authors contributed equally to this project.}

\affiliation{Center for Quantum Information and Quantum Control (CQIQC),
Department of Physics and Department of Electrical \& Computer Engineering,
University of Toronto, Toronto, Ontario, M5S 3G4, Canada}

\date{\today}

\begin{abstract}
In this article we obtain new results for the task of converting a \textit{single} $N$-qubit W-class state (of the form $\sqrt{x_0}\ket{00...0}+\sqrt{x_1}\ket{10...0}+...+\sqrt{x_N}\ket{00...1}$) into maximum entanglement shared between two random parties.  Previous studies in random distillation have not considered how the particular choice of target pairs affects the transformation, and here we develop a strategy for distilling into \textit{general} configurations of target pairs.  We completely solve the problem of determining the optimal distillation probability for all three qubit configurations and most four qubit configurations when $x_0=0$.  Our proof involves deriving new entanglement monotones defined on the set of four qubit W-class states.  As an additional application of our results, we present new upper bounds for converting a generic W-class state into the standard W state $\ket{W_N}=\sqrt{\frac{1}{N}}\left(\ket{10...0}+...+\ket{00...1}\right)$.
\end{abstract}

\maketitle

\section{Introduction}

In quantum information processing, the two-qubit EPR state $\ket{\Phi}=\sqrt{\frac{1}{2}}\left(\ket{00}+\ket{11}\right)$ provides a key resource for performing non-classical tasks such as teleportation \cite{Bennett-1993a} and super-dense coding \cite{Bennett-1992b}.  Thus, for a multi-partite state $\ket{\varphi}_{1...N}$, it is important to know the optimal ways in which EPR entanglement can be obtained between two parties without having to introduce any more entanglement into the system.  This latter constraint is known as the LOCC constraint because it requires each party to perform only local quantum operations (LO) while coordinating their operations through classical communication (CC). 

In general, the optimal conversion of $\ket{\varphi}_{1...N}$ into bipartite entanglement depends on which two final parties are left sharing the entanglement.  One scenario to consider is when two specific parties are designated as the target pair, and the transformation is considered a success if and only if these two parties end up sharing the state $\ket{\Phi}$.  A transformation of this sort is known as a \textbf{specified-pair distillation}.  In this setting, an important problem is to determine the greatest success probability $p_{ij}$ for which the conversion
\[\ket{\varphi}_{1...N}\to\ket{\Phi^{(ij)}}\]
is possible by LOCC.  Here, $\ket{\Phi^{(ij)}}$ denotes an EPR pair between parties $i$ and $j$, and it is assumed that all other parties are in some product state.  While no full solution to this problem is known, some partial results exist \cite{Gour-2005a, Chitambar-2010a}.  

\begin{figure}[b]
\includegraphics[scale=0.6]{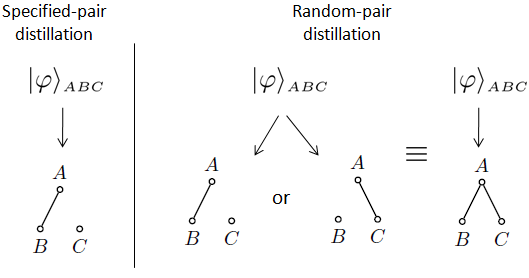}
\caption{\label{RDABACtutorial}
A specified-pair versus random-pair distillation.  For random distillations, it is convenient to combine all the desired outcomes into one configuration graph $\mathcal{G}=(V,E)$ whose edge set encodes the target pairs.  Here, the target pairs are AB and AC.  The ``$\equiv$'' indicates equivalent representations.} 
\end{figure}

A more general question can be posed by allowing the two EPR-entangled parties to vary across the different outcomes in the transformation (see Fig. \ref{RDABACtutorial}).  Any transformation of this form is known as a \textbf{random-pair distillation} (or just simply \textit{random distillation}) because the final two entangled parties are \textit{a priori} unspecified.  Additional constraints to the problem can be added by demanding that the possible target pairs be limited only to some particular subset of all possible pairs.  For example, in the random distillation of Fig. \ref{RDABACtutorial}, the transformation is considered a success only if AB or AC obtain an EPR pair, and not if BC become EPR entangled.  For an $N$-party system, a random distillation can be written as
\begin{equation}
\label{Eq:trans}
\ket{\varphi}_{1...N}\to\{p_{ij},\ket{\Phi^{(ij)}}\}_{(i,j)\in E}
\end{equation}
where $p_{ij}$ is the probability of obtaining $\ket{\Phi^{(ij)}}$ and $E\subset [N]\times [N]$ is some designated set of target bipartite pairs.  The transformation is considered a success if an EPR state is obtained by any pair in $E$.

A convenient way to represent random distillations is through a \textbf{configuration graph} $\mathcal{G}=(V,E)$.  Each party $k$ is identified with a node $v_k\in V$, and an edge $e_{jk}\in E$ connects $v_j$ and $v_k$ if parties $j$ and $k$ form a desired target pair in the distillation (see Fig. \ref{RDABACtutorial}).  It should be emphasized that we are strictly dealing with a single copy of $\ket{\varphi}_{1...N}$, and each edge corresponds to one possible outcome.  Variations to this question in the asymptotic regime have been studied elsewhere \cite{Smolin-2005a,Yang-2009a}.  Given some graph $\mathcal{G}$ and initial state $\ket{\varphi}_{1...N}$, the greatest success probability is given by:
\begin{equation}
P(\varphi,\mathcal{G}):=\sup\sum_{(i,j)\in E} p_{ij}
\end{equation}
where the supremum is taken across all LOCC protocols.

\begin{figure}[b]
\includegraphics[scale=0.6]{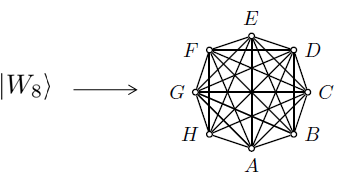}
\caption{\label{RDcomplete}
An $N=8$ example of the ``complete-type'' distillations considered by Fortescue and Lo in Ref. \cite{Fortescue-2007a}.  Such a transformation is a success if \textit{any} two parties become EPR entangled, and this can be achieved with a probability arbitrarily close to one.  Previous research has not considered more general types of configuration graphs than this.} 
\end{figure}

The subject of single-copy random distillation was first introduced and subsequently studied by Fortescue and Lo \cite{Fortescue-2007a, Fortescue-2008a, Fortescue-2009a}.  One prominent finding of their work is that random distillations are, in general, strictly more powerful than specified-pair distillations.  Perhaps the most dramatic example of this is the $N$-qubit state $\ket{W_N}=\sqrt{\frac{1}{N}}\left(\ket{10...0}+\ket{01...0}+...+\ket{00...1}\right)$ and its random distillation into EPR pairs shared between any two parties (see Fig. \ref{RDcomplete}).  In terms of the terminology introduced above, the initial state is $\ket{W_N}$, and the configuration graph is complete (each node connected to every other) such that the conversion is a success if any two parties become EPR entangled.  Fortescue and Lo were able to show that this transformation can be completed with probability arbitrarily close to one \cite{Fortescue-2007a}.  On the other hand, for any two parties, the optimal specified-pair distillation probability is $2/N$. 

In this article, we turn to one largely unexplored question in Fortescue and Lo's work which is the random distillation to \textit{general} configuration graphs, and not just complete graphs.  Specifically, we consider how the target configuration graph affects the random distillation in terms of overall success probability as well as the actual LOCC protocol the parties implement during the transformation.  For example, one particular problem we are able to solve is the four qubit random distillation depicted in Fig. \ref{RDABCDtutorial} which was left as an open problem in Ref. \cite{Fortescue-2009a, Oppenheim-2011a}.

\begin{figure}[t]
\includegraphics[scale=0.6]{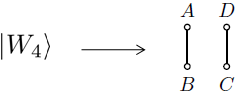}
\caption{\label{RDABCDtutorial}
In Section \ref{Sect:fourqubits} we show that the optimal LOCC probability of achieving this transformation is $2/3$, thus resolving an open problem in Refs. \cite{Fortescue-2009a, Oppenheim-2011a}.  The initial state is $\ket{W_4}=1/2(\ket{1000}+\ket{0100}+\ket{0010}+\ket{0001})$.} 
\end{figure} 

Our focus is on the single-copy random distillation of $N$-party W-class of states, which is the collection of all states reversibly obtainable from $\ket{W_N}$ with a nonzero probability by LOCC.  The choice to limit investigation to this class of states is motivated by multiple factors.  First, from an experimental perspective, W-type entanglement seems relatively easier to generate than other forms of multipartite entanglement, with the state $\ket{W_4}$ already being realized in the laboratory \cite{Wieczorek-2009a}.  In $N>4$ qubit systems, setups have also been proposed for the production of W-class states \cite{Bastin-2009a}.  And for the particular task of random distillation, Fortescue has devised an experimental implementation of W-type random distillation using currently available technology, e.g. ion trap quantum computers \cite{Fortescue-2009a}.  Second, as we will see in the next section, W-class states have a very simple structure which allows us to carefully analyze their behavior under LOCC evolution.  Finally, a large amount of previous research conducted by Fortescue and Lo on random distillations involved W-class states.  Thus, there is an established point of comparison for new results on the subject.   

We summarize our results and outline the structure of this article.
\begin{itemize}
\item In Section \ref{Sect:FLNotation}, we begin by reviewing the results of the Fortescue-Lo Protocol and a described generalization, as well as some related work by Kinta\c{s} and Turgut on entanglement transformations within the W-class \cite{Kintas-2010a}.
\item In Section \ref{Sect:LPO} we construct the ``Least Party Out'' Protocol which distills an arbitrary $N$-qubit W-class state given some target configuration $\mathcal{G}$.  The protocol is similar in nature to the Fortescue-Lo Protocol but we show it to be strictly stronger. 
\item In Section \ref{Sect:fourqubits}, we apply our protocol to three and four qubit systems to obtain the main results of the article.  Every possible three and four qubit configuration graph is considered, and we introduce new four qubit entanglement monotones to show that our protocol is optimal in most cases when $x_0=0$.  
\item In Section \ref{Sect:standardW}, we further apply our results to study the transformation $\ket{\varphi}_{1...N}\to\ket{W_N}$ where $\ket{\varphi}$ is a generic W-class state.  New upper bounds on the optimal conversion probability are obtained.  
\item In the Conclusion, we return to the question of LOCC versus separable operations investigated more heavily in our companion paper \cite{Chitambar-2011b}.  Throughout this article, we will also recall a few other results from that paper.
\end{itemize}

\section{Previous Results and Notation}
\label{Sect:FLNotation}

\subsection*{The Generalized Fortescue-Lo Protocol}

In Ref. \cite{Fortescue-2007a}, Fortescue and Lo developed a protocol which randomly distills the state $\ket{W_N}=\sqrt{\frac{1}{N}}\left(\ket{10...0}+\ket{01...0}+...+\ket{00...1}\right)$ according to a complete configuration graph with success probability arbitrarily close to one (see Fig. \ref{RDcomplete}).  We briefly review the case when $N=3$.  For some $\epsilon>0$, the parties locally perform the measurement given by $M_1=diag[\sqrt{1-\epsilon},1]$ and $M_2=diag[\sqrt{\epsilon},0]$.  If all parties obtain outcome ``1'', the final state is the original state $\ket{W_3}$.  The parties then repeat the same measurement again.  On the other hand if only two parties obtain outcome ``1'', then this pair is left EPR entangled.  But, if one or fewer parties obtain outcome ``1'', all entanglement is destroyed and transformation is a failure.  With the possible recursive step, this protocol can continue for an indefinite number of measurement rounds.  In the end, the total probability of obtaining some EPR pair is $1-O(\epsilon)$.  For $N>3$, the protocol generalizes and likewise the probability of success is $1-O(\epsilon)$.  Here, the probabilities are distributed equally among all possible pairs; that is, with probability $\binom{N}{2}^{-1}-O(\epsilon)$, any two parties $i$ and $j$ obtain an EPR pair.

In Ref. \cite{Fortescue-2009a}, Fortescue briefly considered the problem of applying their protocol to more general configuration graphs, but only a limited discussion is given.  Nevertheless, for a general outcome configuration graph $\mathcal{G}=(V,E)$, we can here describe an obvious way to apply the Fortescue-Lo Protocol.  Starting with the state $\ket{W_N}$, it is converted with equal probability into the $\binom{N}{N-1}$ different $\ket{W_{N-1}}$ states.  Here the difference between these states lies in which of the $N$ parties are entangled.  If all the entangled parties in a particular $\ket{W_{N-1}}$ state are connected according to the graph $\mathcal{G}$, then the state is broken into EPR pairs with probability arbitrarily close to one.  Otherwise, it is converted into the $\binom{N-1}{N-2}$ different $\ket{W_{N-2}}$ states.  This process continues until $\ket{W_3}$ states are obtained.  Either all these parties sharing the $\ket{W_3}$ state are connected in $\mathcal{G}$ or at most two are.  In the former case, EPR pairs are obtained with probability $\approx 1$ whereas in the former, the distillation can be completed with probability $2/3$.  

We will let $P_{FL}(W_N,\mathcal{G})$ denote the distillation success probability of this Generalized Fortescue-Lo Protocol for some configuration $\mathcal{G}$.  Obviously $P(W_N,\mathcal{G})\geq P_{FL}(W_N,\mathcal{G})$.  The ``Least Party Out'' protocol described in the next section will be able to obtain a greater success probability than $P_{FL}(W_N,\mathcal{G})$ in general, and thus tighten the lower bound on $P(W_N,\mathcal{G})$.

\subsection*{Additional notation and the K-T Monotones}

In an $N$-partite system, if a ``standard'' W state $\ket{W_M}$ is shared among parties $S\subset [N]:=\{1,2,...,N\}$ with $|S|=M$, we will often explicate this by writing $\ket{W^{(S)}_{|S|}}$.  Equivalently we can write this state as $\ket{W_{|S|}^{(\overline{T})}}$ where $T=[N]\setminus S$.  Also, we define
\[\ket{W_2^{(ij)}}:=\sqrt{\tfrac{1}{2}}\left(\ket{01}_{ij}+\ket{10}_{ij}\right)\] 
which is local unitarily (LU) equivalent to $\ket{\Phi^{(ij)}}$. 

We will often represent a generic W-class state $\sqrt{x_0}\ket{00...0}+\sqrt{x_1}\ket{01...0}+...+\sqrt{x_n}\ket{00...1}$ by an $N$-component vector:
\begin{align}
\vec{x}=(x_1,&x_2,...,x_N)\notag\\
&\updownarrow\notag\\
\sqrt{x_0}\ket{00...0}+\sqrt{x_1}&\ket{10...0}+...+\sqrt{x_n}\ket{00...1},
\end{align} and $x_0=1-\sum_{i=1}^Nx_i$.  More importantly, even after a basis change - $\ket{0}\to\ket{0'}$ and $\ket{1}\to\ket{1'}$ - the component values $\sqrt{x_i}$ always remain unchanged for $N\geq 3$ \cite{Kintas-2010a}.  When $N=2$, uniqueness can be ensured by demanding that $x_0=0$ and $x_1\geq x_2$.  Therefore, for any number of parties, the vector $\vec{x}$ uniquely specifies the state up to an LU transformation.  For the state $\vec{x}$, we denote
\[x_{n_1}=\max_{1\leq k\leq N} x_k.\]

By disregarding LU transformations and decomposing a general measurement into a sequence of binary outcome POVMS \cite{Anderson-2008a}, we can assume that a local measurement by party $k$ consists of two upper triangular matrices $\{M^{(k)}_1, M^{(k)}_2\}$ whose entries are
\begin{align}
\label{Eq:constraints}
M^{(k)}_1&=\begin{pmatrix}\sqrt{a_1} & b_1  \\0 & \sqrt{c_1}\end{pmatrix} & M^{(k)}_2&=\begin{pmatrix}\sqrt{a_2} & b_2  \\0 & \sqrt{c_2}\end{pmatrix}
\end{align}
with $a_1+a_2=1$ and $c_1+c_2\leq 1$, where equality is achieved by the latter if and only if $M^{(k)}_1$ and $M^{(k)}_2$ are both diagonal.  It is easy to see that this measurement by party $k$ on state $\sqrt{x_0}\ket{00...0}+\sqrt{x_1}\ket{10...0}+...+\sqrt{x_N}\ket{00...1}$ will transform the components as:
\begin{align}
\label{Eq:components}
x_k&\to\frac{c_\lambda}{p_\lambda}x_k, &x_j&\to\frac{a_\lambda}{p_\lambda}x_j\;\;\;1\leq j\not=k\leq N,
\end{align}
with $p_\lambda$ being the probability that outcome $\lambda$ occurs.  From this it is easy to see the following,
\begin{align}
\label{Eq:KTmono}
x_0&\leq\sum_\lambda p_\lambda x_{\lambda,0}& x_i&\geq\sum_\lambda p_\lambda x_{\lambda,i}
\end{align}
for all $1\leq i\leq N$.
We will refer to these as the K-T monotones after Kinta\c{s} and Turgut who first proved the inequalities \cite{Kintas-2010a}.

\section{The ``Least Party Out'' Protocol}
\label{Sect:LPO}
Here we describe our W-class random distillation protocol for a given configuration graph $\mathcal{G}$.  It's called the ``Least Party Out'' (LPO) protocol because it involves systematically removing parties from the $N$-party entanglement with a probability that decreases according to the number of edges connected to the party's node in $\mathcal{G}$.  For some group of parties $S$, we let $\mathcal{G}\setminus S$ denote the subgraph of $\mathcal{G}$ obtained by removing the nodes corresponding to the parties in $S$.

Our protocol can be divided into three phases.  Phase I takes a generic W-class state $\vec{x}$ and converts it into a state $\vec{x'}$ such that $x'_0=0$.  Phase II converts an $x_0=0$ W-class state into standard W states $\ket{W_{|S|}^{(S)}}$ for $2\leq |S|\leq N$ using an ``equal or vanish'' (e/v) measuring scheme.  Phase III then converts the standard W states into EPR pairs given by the configuration graph $\mathcal{G}$.  Phase III is largely inspired by the Fortescue-Lo Protocol in that it involves an indefinite round measurement procedure: each party performs a measurement which, with some probability, leaves the state invariant and thus subject to a repeated round of identical measurement, which again leaves the state invariant with some probability, etc.

\medskip

\noindent\textbf{Phase I: Remove $x_0$ component:}
Input $(\vec{x},\mathcal{G})$ where $\vec{x}$ is an $N$-partite W-class state and $\mathcal{G}$ is some configuration graph with $N$ nodes.  If $x_0=0$, proceed to Phase II.  Otherwise, choose some party $n_1$ with the largest component value to perform the measurement \eqref{Eq:constraints} with $a_1=c_1=\lambda$, $b_1=-\sqrt{\lambda\frac{x_0}{x_1}}$, and $c_2=0$.  The values for $a_2$, $b_2$ and $\lambda$ are fixed by the measurement being complete \cite{Cui-2010a}.  Outcome ``1'' occurs with probability
\begin{equation}
\label{Eq:x0prob}
\frac{2x_{n_1}(1-x_0)}{x_0+2x_{n_1}+\sqrt{x_0^2+4x_{n_1}x_0}}.
\end{equation} 
and the resultant state has no zeroth component.  For outcome ``2'', the state is either a product state, in which case the protocol halts as a failure, or the state is entangled with party $n_1$'s component being zero and the state still having a zeroth component.  In both cases, redefine $\vec{x}$ as the post-measurement state, but set $\mathcal{G}$ as $\mathcal{G}\setminus n_1$ only after outcome ``2''.  Repeat Phase I with input $(\vec{x},\mathcal{G})$.

\medskip

\noindent\textbf{Phase II: ``Equal or Vanish'' (e/v) Subroutine:}  Input $(\vec{x},\mathcal{G})$ where $x_0=0$ and $\vec{x}$ is shared between $2\leq|S|\leq N$ parties.
\begin{itemize}
\item[(1)] If there does not exist an isolated node $v_k$ in $\mathcal{G}$ (one without any outgoing edges), proceed to the next step (2).  Otherwise, when $|S|=2$ the protocol halts as a failure, and when $|S|>2$, the ``isolated'' party ``k'' performs the dis-entangling measurement $M^{(k)}_1=diag[1,0]$ and $M^{(k)}_2=diag[0,1]$.  If outcome ``1'' occurs, redefine $\vec{x}$ as the post-measurement state and set $\mathcal{G}$ as $\mathcal{G}\setminus k$; repeat the e/v subroutine on input $(\vec{x},\mathcal{G})$.  If outcome ``2'' happens, the protocol terminates as a failure.
\item[(2)]  If every component in $\vec{x}$ is maximal, then $\vec{x}$ is a standard W state $\ket{W_{|S'|}^{(S')}}$ and proceed to Phase III.  Otherwise, choose some party $k$ such that (i) $x_k$ is non-maximal, and (ii) party $k$ is connected to some party whose component is maximal.  If no party satisfies both these conditions, then choose some party $k$ which just satisfies condition (i).  He/she then performs a two-outcome measurement with operators $M_1^{(k)}=diag[\sqrt{\tfrac{x_{k}}{x_{n_1}}},1]$ and $M_2^{(k)}=diag[\sqrt{1-\tfrac{x_{k}}{x_{n_1}}},0]$.  Party $k$'s component value equals the maximum upon outcome ``1'' and vanishes upon outcome ``2''.  In both cases, redefine $\vec{x}$ as the post-measurement state, but set $\mathcal{G}$ as $\mathcal{G}\setminus k$ only after outcome ``2''.  Repeat the e/v subroutine on the new input $(\vec{x},\mathcal{G})$ (see Fig. \ref{evfigtree}).
\end{itemize}
For Phase II input $(\vec{x},\mathcal{G})$, the final success states of the e/v subroutine are $\ket{W^{(S')}_{|S'|}}$ for $2\leq |S'|\leq |S|$ where for each party $*$ such that $x_*=x_{n_1}$, either $*\in S'$ or no party in $S'$ is connected to $*$ in $\mathcal{G}$.  The latter case occurs when there is only one party with a maximum component and all parties connected to $*$ measure a ``vanish'' outcome; consequently, $*$ becomes an isolated party and removes itself from the system via step (1) above.  Let
\begin{align}
\lambda_{\vec{x},\mathcal{G}}(W_{|S'|}^{(S')}):=&\text{ the probability or obtaining $\ket{W^{(S')}_{|S'|}}$}\notag\\ &\text{ via the e/v subroutine for input $(\vec{x},\mathcal{G})$.}
\end{align}
Note that $\lambda_{\vec{x},\mathcal{G}}$ is a smooth function of the component values $x_i$ and can be explicitly computed from the measurement operators given above.  For example, $\lambda_{\vec{x},\mathcal{G}}(W_N)=N\frac{\prod_{k\not=n_1}x_k}{x_{n_1}^{N-2}}$.  Also, $\lambda_{W_N,\mathcal{G}}(W_M)=\delta_{MN}$.

\begin{figure}[t]
\includegraphics[scale=0.6]{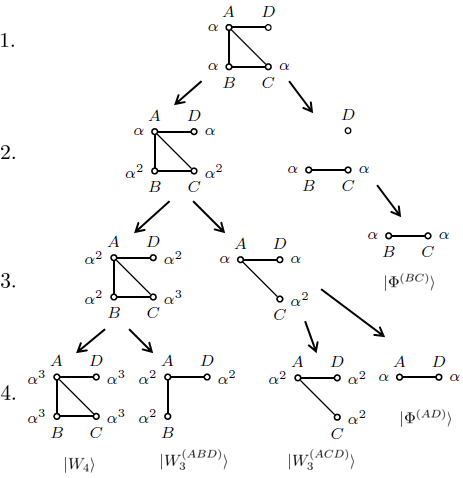}
\caption{\label{evfigtree}
``Equal or Vanish'' Subroutine (Phase II) for the normalized state $\frac{1}{1+3\alpha}(\alpha,\alpha,\alpha,1)$ and the configuration graph with edges $\{AB,AC,AD, BC\}$.  1.  David's component is largest and Alice is a connected to him with a lesser component value.  She performs an e/v measurement.  2.  For outcome ``vanish'' (right branch) she is separated from the system, and since David is not connected to either Bob or Charlie, he immediately removes himself from the system leaving $\ket{\Phi^{(BC)}}$ with some probability.  For outcome ''equal'' (left branch) the components of all other parties receive a factor of $\alpha$, and Alice's component is now maximum equaling David's.  Bob is a connected party to Alice with a lesser component value and he performs an e/v measurement.  3.  Again, either Bob vanishes (right branch) or all other components except his receive a factor of $\alpha$.  In both cases, Charlie is then a connected party to Bob with a lesser component value and he performs an e/v measurement.  4.  The final outcome states along these branches are $\ket{W_4}$, $\ket{W_3^{(ABD)}}$, $\ket{W_3^{(ACD)}}$, and $\ket{\Phi^{(AD)}}$.}
\end{figure}

\medskip

\noindent\textbf{Phase III: Obtaining EPR Pairs:}  Input $(W_{|S|}^{(S)},\mathcal{G})$ with $\mathcal{G}$ having $|S|$ nodes and at least one edge.  If $|S|=2$, the state is an EPR pair and protocol halts as a success.  Otherwise, Phase III of the protocol is defined recursively such that for $|S|>2$, the procedure depends on a pre-defined random distillation protocol for $\ket{W_{|S'|}^{(S')}}$ with $|S'|<|S|$.  Let  
\begin{align}
P_{III}(W_{|S|}^{(S)},\mathcal{G}):=&\text{ the probability of distilling $\ket{W_{|S|}^{(S)}}$}\notag\\
&\text{  into $\mathcal{G}$ via the Phase III procedure.}\notag
\end{align}
If $|S|=2$, set $P_{III}(W_2,\mathcal{G})=1$ by definition.  For $|S|>2$, identify the party $k$ whose node in $\mathcal{G}$ has a least number of connected edges.  He/she performs the measurement with operators given by $M_1^{(k)}=diag[\sqrt{\alpha},1]$ and $M_2^{(k)}=diag[\sqrt{1-\alpha},0]$ where $\alpha$ is determined according to the discussion below Eq. \eqref{Eq:probG}.  Outcome ``2'' occurs with probability $(1-\alpha)\frac{|S|-1}{|S|}$ and the resultant state is $\ket{W_{|S|-1}^{(\overline{k})}}$.  Phase III is then repeated on this state and the reduced graph $\mathcal{G}\setminus k$.  

\begin{figure}[t]
\includegraphics[scale=0.6]{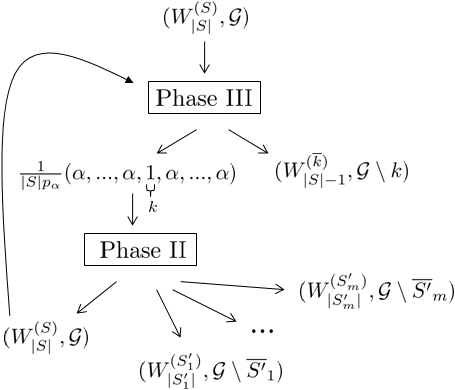}
\caption{\label{Phase3fig}
Phase III receives an input state $W_{|S|}^{(S)}$ and a configuration graph $\mathcal{G}$.  Party $k$ performs an e/v measurement.  One outcome is a standard $W$ state with party $k$ removed, and the other is the state $\tfrac{1}{|S|p_\alpha}(\alpha,...,\alpha,1,\alpha,...,\alpha)$.  Phase II is applied on this state outputting either W states or a product (failure) state.  Phase III will next be initiated on each of the W states, and for any W state $W_{|S'|}^{(S')}$ with $|S'|<|S|$, the transformation success probability from this point onward is given by $P_{III}(W_{|S'|}^{(S')},\mathcal{G}\setminus \overline{S'})$; this value is already known by recursion.  However, for the state $W_{|S|}^{(S)}$, performing Phase III again will generate an indefinite loop, but one whose overall success probability converges to $\frac{f(\alpha)}{1-\alpha^{|S|-1}}$ (see Eqs. \eqref{Eq:falpha} and \eqref{Eq:probG}).}
\end{figure}

Outcome ``1'' happens with probability $p_\alpha=\frac{1+\alpha(|S|-1)}{|S|}$ and the post-measurement state is $\vec{y}_\alpha$, which (up to a permutation between party 1 and $k$) takes the form: $\tfrac{1}{|S|p_\alpha}(1,\alpha,\alpha,...,\alpha)$.  Party $k$ then has the largest component value in $\vec{y}_\alpha$, and the e/v subroutine (Phase II) is performed on the input $(\vec{y}_\alpha,\mathcal{G})$.  The e/v subroutine will either output the states $\ket{W^{(S')}_{|S'|}}$ where $|S'|<|S|$ with respective probabilities $\lambda_{\vec{y}_\alpha,\mathcal{G}}(W^{(S')}_{|S'|})$, or the original state $\ket{W_{|S|}^{(S)}}$ with probability $\alpha^{|S|-1}$.  In the former case, the Phase III procedure is performed on the input $(W^{(S')}_{|S'|},\mathcal{G}\setminus \overline{S'})$.  Accounting for all states with $|S'|<|S|$, their total distillation success probability is
\begin{align}
\label{Eq:falpha}
&f(\alpha) =(1-\alpha)\tfrac{|S|-1}{|S|}\cdot P_{III}(W_{|S|-1}^{(\overline{k})},\mathcal{G}\setminus k)\notag\\
&+\tfrac{1+\alpha(|S|-1)}{|S|}\sum_{2\leq |S'|<|S|}\lambda_{\vec{y}_\alpha,\mathcal{G}}(W^{(S')}_{|S'|}) P_{III}(W^{(S')}_{|S'|},\mathcal{G}\setminus \overline{S'}),
\end{align}
where the sum is taken over all subsets $S'$ such that $2\leq |S'|<|S|$ and either $k\in S'$ or no party in $S'$ is connected to $k$ in $\mathcal{G}$.  If the e/v subroutine outputs the original $\ket{W_{|S|}^{(S)}}$, repeat Phase III again on the same input $(W_{|S|}^{(S)},\mathcal{G})$.  This will generate an indefinite loop in which for each cycle, the probability of distillation success is $f(\alpha)$, and the probability of continuing on for another cycle is $\alpha^{|S|-1}$.  Therefore, the total success probability across all cycles is given by the geometric sum $\sum_{r=0}^\infty[\alpha^{|S|-1}]^rf(\alpha)$.  To maximize this value, we set

\begin{align}
\label{Eq:probG}
P_{III}(W_{|S|}^{(S)},\mathcal{G})&=\sup_{0\leq \alpha< 1}\frac{f(\alpha)}{1-\alpha^{|S|-1}}.
\end{align}

This determines the original value of $\alpha$ in the Phase III measurement operators: if \eqref{Eq:probG} obtains its supremum in the interval $[0,1)$, then $\alpha$ is chosen to be any of these critical points; if the supremum is obtained in the limit $\alpha\to 1$, then  $\alpha=1-\epsilon$ for any desired $\epsilon>0$.  The smaller the value of $\epsilon$, the closer the success probability approaches $P_{III}(W_{|S|}^{(S)},\mathcal{G})$.  Observe that when the supremum is obtained at $\alpha=0$, the first measurement by party $k$ will deterministically dis-entangle the party from the rest of the system.

It is also important to note that the optimization of \eqref{Eq:probG} can always be efficiently performed.  By the recursive construction, the values for $P_{III}(W_{|S|-1}^{(\overline{k})},\mathcal{G}\setminus k)$ and $P_{III}(W^{(S')}_{|S'|},\mathcal{G}\setminus \overline{S'})$ are just real numbers known \textit{a priori}.  Furthermore, the functions $\lambda_{\vec{y}_\alpha,\mathcal{G}}(W^{(S')}_{|S'|})$ are smooth functions of $\alpha$, and thus, Eq. \eqref{Eq:probG} represents a single-variable smooth function whose supremum value can be easily computed.  In total then, for a state $N$-partite state $\vec{x}$ with $x_0=0$ and configuration graph $\mathcal{G}$, the overall success probability of the LPO protocol is given by 
\begin{equation}
\label{Eq:LPOprob}
P_{LPO}(\vec{x},\mathcal{G}):=\sum_{2\leq |S'|\leq N}\lambda_{\vec{x},\mathcal{G}}(W^{(S')}_{|S'|})\cdot P_{III}(W^{(S')}_{|S'|},\mathcal{G}\setminus \overline{S'})
\end{equation}
where the sum is taken over all subsets $S'$ such that $2\leq |S'|\leq N$ and for each party $*$ such that $x_*=x_{n_1}$, either $*\in S'$ or no party in $S'$ is connected to $*$ in $\mathcal{G}$.

\section{Main results: the LPO Protocol on three and four qubits}
\label{Sect:fourqubits}

\noindent{\textbf{Summary of Results:}}  Before working through the LPO Protocol in detail on three and four qubit systems, we summarize the overall results.  For three qubits, the possible configuration graphs are depicted in Fig. \ref{RDABAC}, and upper bounds on the transformation success probabilities are given by Eqs. \eqref{Eq:RDABAC} and \eqref{Eq:ABACBC} respectively.  In both cases, when $x_0=0$ these bounds can be approached arbitrarily close.

For four qubits, there are six different families of configurations depicted in Figs. \ref{RDABACAD} -- \ref{RDAll}.  For states with $x_0=0$, we have completely solved the random distillation problem for all configurations except VI.  Upper bounds on configurations I -- V are given by Eqs. \eqref{Eq:RDABACAD} --  \eqref{Eq:kappa} respectively.  

\subsection*{Three qubits}

\begin{figure}[htbp]
\includegraphics[scale=0.6]{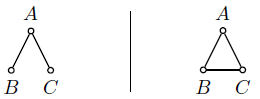}
\caption{\label{RDABAC}
(Left) Configuration $\mathcal{G}_{\wedge}$.  (Right) Configuration $\mathcal{G}_{\triangle}$.  An upper bound on the success probability is given by Eqs. \eqref{Eq:RDABAC} and \eqref{Eq:ABACBC} respectively which is effectively tight when $x_0=0$.  For $\ket{W_3}$, these probabilties are $2/3$ and $1$ respectively.} 
\end{figure} 

As a first example of the LPO protocol, consider the state $\ket{W_3}$ and the configuration graph given by $\mathcal{G}_{\wedge}$ in Fig. \ref{RDABAC}.  In Phase III of the protocol, we can choose the ``least party'' to be either Bob or Charlie (say it's Charlie).  He performs a measurement as described above, and either $\ket{\Phi^{(AB)}}$ is obtained or the post-measurement state is $\vec{y}_\alpha=\tfrac{1}{2\alpha+1}(\alpha,\alpha,1)$.  For the latter, the e/v subroutine obtains $\ket{\Phi^{(AC)}}$ with probability $\lambda_{\vec{y}_\alpha,\mathcal{G}_\wedge}(W_2^{(AC)})=\frac{2\alpha(1-\alpha)}{2\alpha+1}$.  Thus, we have 
 $f(\alpha)=\frac{2}{3}(1-\alpha)+\frac{2}{3}(1-\alpha)\alpha$ and therefore $\frac{f(\alpha)}{1-\alpha^{2}}$ takes a constant value of $\frac{2}{3}$.  Hence, $\alpha$ can be chosen as $0$ in Charlie's measurement and 
\begin{equation}
\label{Eq:W3wedge}
P_{III}(W_3,\mathcal{G}_\wedge)=\frac{2}{3}.
\end{equation}  
 
For a more general state $\vec{x}=(x_A,x_B,x_C)$ with $x_A\geq x_B\geq x_C$ and $x_0=0$, we have $\lambda_{\vec{x},\mathcal{G}_\wedge}(W^{(AB)}_{2})=2x_B(1-\frac{x_C}{x_A})$, $\lambda_{\vec{x},\mathcal{G}_\wedge}(W^{(AC)}_{2})=2x_C(1-\frac{x_B}{x_A})$, and $\lambda_{\vec{x},\mathcal{G}_\wedge}(W_{3})=3\frac{x_Bx_C}{x_A}$.  Therefore, by Eq. \ref{Eq:LPOprob}, the distillation probability is given by 
\begin{align}
\label{Eq:RDABAC}
P_{LPO}(\vec{x},\mathcal{G}_\wedge)&=2x_B(1-\frac{x_C}{x_A})+2x_C(1-\frac{x_B}{x_A})+2\frac{x_Bx_C}{x_A}\notag\\
&=2x_B+2x_C-2\frac{x_Bx_C}{x_A}.
\end{align}
One might wonder if this probability is optimal. It turns out that the answer is yes. See Eq. \eqref{Eq:RDABBDAD} below and the discussion there. 

On the other hand, consider configuration $\mathcal{G}_\triangle$ given in Fig. \ref{RDABAC}.  We can still choose Charlie as the ``least'' party, and this time, the possible EPR success states of the e/v subroutine are $\ket{\Phi^{(AC)}}$ and $\ket{\Phi^{(BC)}}$.  We have $\lambda_{\vec{y}_\alpha,\mathcal{G}_\triangle}(W_2^{(AC)})=\lambda_{\vec{y}_\alpha,\mathcal{G}_\triangle}(W_2^{(AC)})=\frac{2\alpha(1-\alpha)}{2\alpha+1}$ and so $f(\alpha)=\frac{2}{3}(1-\alpha)+\frac{4}{3}(1-\alpha)\alpha$.  Thus,
\[P_{LPO}(W_3,\mathcal{G}_\triangle)=\sup_{0\leq \alpha< 1}\frac{2}{3}\frac{(1-\alpha)(1+2\alpha)}{1-\alpha^{2}}=1.\]
This value is obtained in the limit $\alpha\to 1$ which means the LPO protocol calls for infinitesimal measurements with $\alpha=1-\epsilon$.  Hence, for three qubits, the LPO protocol reduces to the Fortescue-Lo Protocol for distilling the state $\ket{W_3}$.  

Since $P_{III}(W_3,\mathcal{G}_\triangle)=1$ for the three-edge configuration, when considering the state $\vec{x}=(x_A,x_B,x_C)$ with $x_A\geq x_B\geq x_C$ and $x_0=0$, we obtain the distillation probability    
\begin{align}
\label{Eq:ABACBC}
P_{LPO}(\vec{x},\mathcal{G}_\triangle)=&2x_B(1-\frac{x_C}{x_A})+2x_C(1-\frac{x_B}{x_A})+3\frac{x_Bx_C}{x_A}\notag\\
&=2x_B+2x_C-\frac{x_Bx_C}{x_A}.
\end{align}
Just as with the configuration graph $\mathcal{G}_\wedge$, this probability is optimal as we will see from Eq. \eqref{Eq:RDABBDAD} below.  We now turn to four qubits where, unlike the two cases just examined, there exists configurations for which $f(\alpha)$ obtains a maximum in the interval $(0,1)$. 

\subsection*{Four qubits}

Next, we apply the LPO protocol to four qubit W-class states.  We will only consider a subset of possible configuration graphs, but any other can be obtained by a permutation of parties.

\medskip

\noindent\textbf{Configurations I (Fig. \ref{RDABACAD}):}

\begin{figure}[htbp]
\includegraphics[scale=0.6]{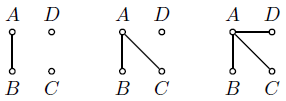}
\caption{\label{RDABACAD}
Let $\mathcal{G}_I$, $\mathcal{G}_I'$ and $\mathcal{G}_I''$ be the first, second and third of the above configurations respectively.  An upper bound on the success probability is given by Eq. \eqref{Eq:RDABACAD} which is effectively tight when $x_0=0$.  For $\ket{W_4}$, this probability is $1/4$ for each configuration.} 
\end{figure} 

For a generic W-class state $\vec{x}=(x_A, x_B, x_C, x_D)$, whenever $x_A<x_j$ for some $j\in \{B,C,D\}$, an upper bound on distilling to any of these configurations is $2 x_{A}$ by the K-T monotones.  However, when $x_A$ is the largest component value, we have
\begin{align}
\label{Eq:RDABACAD}
P(\vec{x},\mathcal{G}_I)&\leq 2 x_{B}\notag\\
P(\vec{x},\mathcal{G}'_I)&\leq 2x_A-2\tfrac{(x_A-x_B)(x_A-x_C)}{x_A}\notag\\
P(\vec{x},\mathcal{G}''_I)&\leq 2x_A-2\tfrac{(x_A-x_B)(x_A-x_C)(x_A-x_D)}{x_A^2}
\end{align}
as proven in Ref. \cite{Chitambar-2011b}.  When $x_0=0$, these are precisely the rates obtained by the LPO Protocol, and so our protocol is optimal for such states.  Note that setting $x_D=0$ proves \eqref{Eq:RDABAC} to be optimal.

\medskip

\noindent\textbf{Configuration II (Fig. \ref{RDABBDAD}):}
\begin{figure}[htbp]
\includegraphics[scale=0.6]{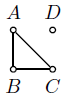}
\caption{\label{RDABBDAD}
Let $\mathcal{G}_{II}$ be the above configuration.  An upper bound on the success probability is given by Eq. \eqref{Eq:RDABBDAD} which is effectively tight when $x_0=0$.  For $\ket{W_4}$, this probability is $3/4$} 
\end{figure} 

For a generic W-class state $\vec{x}=(x_A, x_B, x_C, x_D)$, if we assume without loss of generality that $x_A\geq x_B\geq x_C$, then we have
\begin{equation}
\label{Eq:RDABBDAD}
P(\vec{x},\mathcal{G}_{II})\leq 1-x_0-x_D-\frac{(x_A-x_B)(x_A-x_C)}{x_A},
\end{equation}
as also proven in Ref. \cite{Chitambar-2011b}.  When $x_0=0$, the LPO protocol can approach this upper bound arbitrarily close.  Note that this also proves Eq. \eqref{Eq:ABACBC} to be optimal. 

\medskip

\noindent\textbf{Configurations III (Fig. \ref{RDABBCCDAD}):}
\begin{figure}[htbp]
\includegraphics[scale=0.6]{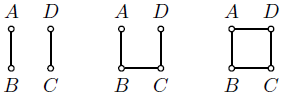}
\caption{\label{RDABBCCDAD}
Let $\mathcal{G}_{III}$ be any of the above configurations.  In each of these, $(A,C)$ and $(B,D)$ are unconnected pairs.   An upper bound on the success probability is given by Eq. \eqref{Eq:tau} which is effectively tight when $x_0=0$.  For $\ket{W_4}$, this probability is $2/3$ for each of these configurations.} 
\end{figure} 

A common feature to all of these configurations is that for each party, there is at least one other party to whom he/she is not connected.  We will refer to such a pair as unconnected.  For example $(A,C)$ and $(B,D)$ form unconnected pairs in each of the above configurations.  We introduce the following entanglement monotones to put an upper bound on the probability for transformations of Configurations III.  For a generic W-class state $\vec{x}=(x_A,x_B,x_C,x_D)$, let $n_1$ be some party whose component is maximum, $n_1'$ the party unconnected to $n_1$ with largest component value, and $p$ and $p'$ the other two parties.  For definitiveness, if party $n_1$ has two unconnected parties (which is possible in the first two of Configurations III), take $p'$ to be the other one besides $n_1'$.  Define the function 
\begin{align}
\tau(\vec{x})&=2x_{p}+2x_{p'}-2\frac{x_{p}x_{p'}}{x_{n_1}}+\frac{2}{3}\frac{x_{p}x_{p'}x_{n_1'}}{x_{n_1}^2}.\notag
\end{align}
Note that $\tau(W_4)=2/3$.
\begin{theorem}
\label{Thm:tau}
The function $\tau$ is an entanglement monotone.
\end{theorem}
\begin{proof}
See Appendix \ref{Apx:tau}
\end{proof}

As a result of this theorem, we have that for a state $\vec{x}$,
\begin{equation}
\label{Eq:tau}
P(\vec{x},\mathcal{G}_{III})\leq \tau(\vec{x}).
\end{equation}
For the LPO Protocol, first consider the initial state $\ket{W_4}$.  In each of the configurations, either party A or D can be chosen as the ``least'' party.  Regardless of the choice, we have 
\[\sum_{(i,j)\in E}\lambda_{\vec{y}_\alpha,\mathcal{G}_{III}}(W_2^{(ij)})=(4\alpha+2\alpha^2)(1-\alpha).\]
Here $E$ denotes the edge set of whatever configuration is considered.  By Eq. \eqref{Eq:W3wedge}, we also know that $P_{III}(W_3^{(BCD)},\mathcal{G}\setminus A)=2/3$.   Thus, 
\[P_{III}(W_4,\mathcal{G}_{III})=\sup_{0\leq \alpha< 1}\frac{\tfrac{1}{2}+\alpha+\tfrac{1}{2}\alpha^2}{1+\alpha+\alpha^2}=\frac{2}{3}\]
which obtains this value as $\alpha\to 1$.  Thus, for any $\epsilon>0$, 
\[2/3-\epsilon<P(W_4,\mathcal{G}_{III})\leq 2/3.\]

When $x=0$ and a state $\vec{x}=(x_A,x_B,x_C,x_D)$ is considered, it is straight forward to compute the probabilities for $\lambda_{\vec{x},\mathcal{G}_{III}}(W_{|S|}^{(S)})$.  Doing so and using Eq. \eqref{Eq:LPOprob} shows that the probability $\tau(\vec{x})$ can be approached arbitrarily close using the LPO protocol. 

\medskip

\noindent\textbf{Configuration IV (Fig. \ref{RDABBDADACBC}):}
\begin{figure}[htbp]
\includegraphics[scale=0.6]{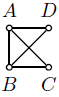}
\caption{\label{RDABBDADACBC}
Let $\mathcal{G}_{IV}$ be the above configuration.  We say two parties are edge complementary if their nodes have a different number of connected edges.  For example, $A$ is edge complementary to both $C$ and $D$.  An upper bound on the success probability is given by Eq. \eqref{Eq:Gamma} which is effectively tight when $x_0=0$.  For $\ket{W_4}$, this probability is $5/6$.} 
\end{figure}

We first consider the transformation of the standard W state $\ket{W_4}$.  The Phase II probabilities are
\begin{align}
\lambda_{\vec{y}_\alpha,\mathcal{G}_{IV}}(W_2^{(AD)})&=\frac{1}{1+3\alpha}\cdot2\alpha(1-\alpha)^2,\notag\\ \lambda_{\vec{y}_\alpha,\mathcal{G}_{IV}}(W_2^{(BD)})&=\frac{1}{1+3\alpha}\cdot2\alpha(1-\alpha)^2,\notag\\ \lambda_{\vec{y}_\alpha,\mathcal{G}_{IV}}(W_3^{(ABD)})&=\frac{1}{1+3\alpha}\cdot3\alpha^2(1-\alpha),\notag\\ 
\lambda_{\vec{y}_\alpha,\mathcal{G}_{IV}}(W_3^{(ACD)})&=\frac{1}{1+3\alpha}\cdot3\alpha^2(1-\alpha),\notag\\
\lambda_{\vec{y}_\alpha,\mathcal{G}_{IV}}(W_3^{(BCD)})&=\frac{1}{1+3\alpha}\cdot3\alpha^2(1-\alpha).\notag
\end{align}
This gives
\begin{equation*}
P_{III}(W_4,\mathcal{G}_{IV})=\sup_{0\leq \alpha< 1}\frac{f(\alpha)}{1-\alpha^3}=\frac{\tfrac{3}{4}+\alpha+\tfrac{3\alpha^2}{4}}{1+\alpha+\alpha^2}=5/6
\end{equation*}
for which the value is obtained as $\alpha\to 1$.

Consider a generic W-class state $\vec{x}=(x_A,x_B,x_C,x_D)$.  We say that a party is edge complementary to a party if there corresponding nodes in $\mathcal{G}_{IV}$ have a different number of connected edges.  For the particular configuration $\mathcal{G}_{IV}$, let $n_1$ denote some party with the largest component, $n_1'$ the party edge complementary to $n_1$ with the largest component, and $e_2$ and $e_3$ the other two parties having 2 and 3 edges respectively.  Define the function:
\begin{align*}
\Gamma(\vec{x})=\begin{cases} 2x_{e_3}+(x_{e_2}+x_{n_1'})\left(2-\frac{x_{e_3}}{x_{n_1}}\right)-2\frac{x_{n_1'}x_{e_2}}{x_{n_1}}\\\hspace{.5cm}+\frac{4}{3}\frac{x_{e_2}x_{e_3}x_{n_1'}}{x_{n_1}^2}\;\;\;\text{if $n_1$ has 3 connected edges},\\
2x_{n_1'}+2x_{e_3}-\frac{x_{n_1'}x_{e_3}}{x_{n_1}}+\frac{x_{e_2}x_{e_3}x_{n_1'}}{3x_{n_1}^2}\;\;\;\text{if $n_1$ has 2}\\
\hspace{.5cm}\text{connected edges}.
\end{cases}
\end{align*}
Note that $\Gamma(W_4)=5/6$.
\begin{theorem}
\label{Thm:gamma}
The function $\Gamma$ is an entanglement monotone.
\end{theorem}
\begin{proof}
The proof is nearly identical in structure to the one given for $\tau$ in Appendix \ref{Apx:tau}.
\end{proof}

As a result, it immediately follows that 
\begin{equation}
\label{Eq:Gamma}
P(\vec{x},\mathcal{G}_{IV})\leq \Gamma(\vec{x}).
\end{equation}
And just as in the case of Configurations III, the LPO protocol can approach this upper bound arbitrarily close whenever $x_0=0$.  Highlighting the standard W state, we have that for any $\epsilon>0$,
\[5/6-\epsilon<P(W_4,\mathcal{G}_{IV})\leq 5/6.\]
It should be noted that 5/6 is the also the optimal transformation probability if one considers transformations within the more general class of separable operations \cite{Chitambar-2011b}.

\medskip

\noindent\textbf{Configuration V (Fig. \ref{RDAll}):}
\begin{figure}[htbp]
\includegraphics[scale=0.6]{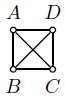}
\caption{\label{RDAll}
Let $\mathcal{G}_{V}$ be the above configuration.   An upper bound on the success probability is given by Eq. \eqref{Eq:kappa} which is effectively tight when $x_0=0$.  For $\ket{W_4}$, this probability is $1$.} 
\end{figure}

For the state $\ket{W_4}$, the Fortescue-Lo Protocol achieves this distillation configuration with probability arbitrarily close to one.  For more general states, we recall the results from Ref. \cite{Chitambar-2011b}:
\begin{equation}
\label{Eq:kappa}
P(\vec{x},\mathcal{G}_{V})\leq 1-x_0-\frac{(x_A-x_B)(x_A-x_C)(x_A-x_D)}{x_A^2}
\end{equation}
where we have assumed without loss of generality that $x_A$ is the largest component value.  The LPO protocol can achieve this probability as close as desired whenever $x_0=0$.

\medskip

\noindent\textbf{Configuration VI (Fig. \ref{RDABBDADAC}):}
\begin{figure}[htbp]
\includegraphics[scale=0.6]{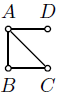}
\caption{\label{RDABBDADAC}
Let $\mathcal{G}_{VI}$ be the above configuration.  For $\ket{W_4}$, the LPO Protocol gives a success probability of $\frac{1}{6}(3+\sqrt{3})$.  We conjecture this to be optimal.} 
\end{figure}

For this configuration, we only work out the Phase III calculation for the standard W state $\ket{W_4}$.  In this case, David is the ``least'' party and he measures first.  Outcome ``2'' is the state $\ket{W^{(ABC)}_3}$ obtained with probability $3/4(1-\alpha)$; from here, we have $P_{III}(W_3^{(ABC)},\mathcal{G}_{VI}\setminus D)=1$.  Outcome ``1'' is the state $\frac{1}{1+3\alpha}(\alpha,\alpha,\alpha,1)$.  The ensuing e/v subroutine is described in Figure \ref{evfigtree}.  We have 
\begin{align}
\lambda_{\vec{y}_\alpha,\mathcal{G}_{VI}}(W_2^{(BC)})&=\frac{1}{1+3\alpha}\cdot2\alpha(1-\alpha),\notag\\ \lambda_{\vec{y}_\alpha,\mathcal{G}_{VI}}(W_2^{(AD)})&=\frac{1}{1+3\alpha}\cdot2\alpha(1-\alpha)^2,\notag\\ \lambda_{\vec{y}_\alpha,\mathcal{G}_{VI}}(W_3^{(ABD)})&=\frac{1}{1+3\alpha}\cdot3\alpha^2(1-\alpha),\notag\\ \lambda_{\vec{y}_\alpha,\mathcal{G}_{VI}}(W_3^{(ACD)})&=\frac{1}{1+3\alpha}\cdot3\alpha^2(1-\alpha).\notag
\end{align}
This gives
\begin{align}
P_{III}(W_4,\mathcal{G}_{VI})=\sup_{0\leq \alpha< 1}\frac{f(\alpha)}{1-\alpha^3}&=\frac{\tfrac{3}{4}+\alpha+\tfrac{\alpha^2}{2}}{1+\alpha+\alpha^2}\notag\\
&=\frac{1}{6}(3+\sqrt{3})
\end{align}
which obtains this maximum when $\alpha=\frac{1}{2}(\sqrt{3}-1)$.  The generalized Fortescue-Lo Protocol gives a rate of 3/4 so we see an improvement in our protocol.  For an upper bound, it is known that this transformation cannot be accomplished with any probability greater than 5/6 by the more general class of separable operations \cite{Chitambar-2011b}.  Thus, we summarize our result by
\begin{align}
P_{FL}(W_4,\mathcal{G}_{VI})<P_{LPO}(W_4,\mathcal{G}_{VI})&\approx\tfrac{3+\sqrt{3}}{6}\notag\\
&\leq P(W_4,\mathcal{G}_{VI})\leq \tfrac{5}{6}.
\end{align}
We use the ``$\approx$'' symbol for the $P_{LPO}$ value since it can be approached arbitrarily close.

While we conjecture that this protocol is optimal for the state $\ket{W_4}$ and the configuration graph $\mathcal{G}_{VI}$, unfortunately it does not appear optimal for more general four qubits states.  Indeed, suppose that we begin with state $\vec{x}=(x_A,x_B,x_C,x_D)$ with $x_0=0$ and $x_A\geq x_B\geq x_C\geq x_D$.  The LPO Protocol says that we should first perform the e/v subroutine with respect to party 1, and then implement Phase III on the state $\ket{W_4}$.  The total probability is then given by Eq. \eqref{Eq:LPOprob}.  Explicitly computing it yields:
\begin{align}
P_{LPO}(\vec{x},\mathcal{G}_{VI})&=2(x_B+x_C+x_D)-\frac{x_Bx_D}{x_A}-2\frac{x_Bx_C}{x_A}\notag\\
&-2\frac{x_Cx_D}{x_A}+\frac{3+2\sqrt{3}}{3}\frac{x_Bx_Cx_D}{x_A^2}.
\end{align}
Now suppose we have $x_A=1-3t$ and $x_B=x_C=x_D=t$ with $t<1/4$.  Since Alice's component is strictly greater than all other components, she can make a weak measurement such that her component value is still the largest in both post-measurement states.  Specifically, when she performs the measurement given by Eq. \eqref{Eq:constraints} with $(a_1,c_1,a_2,c_2)$ in some neighborhood of $(1/2,1/2,1/2,1/2)$, the average change in $P_{LPO}$ is
\[\overline{\Delta P_{LPO}}=-(a-c)^2\cdot\frac{t^2}{1-3t}\left(20-t\frac{12-8\sqrt{3}}{1-3t}\right)\]
which can be positive for $t$ close to $1/4$.  Therefore, $P_{LPO}$ cannot be the optimal probability for the initial state $\vec{x}$ since a weak measurement by Alice increases the overall transformation probability.  

It should be emphasized that for the transformation of $\ket{W_4}$ according to the LPO protocol, we never encounter a state like $\vec{x}$.  The only time Alice's component is larger than David's is after David performs an e/v measurement and his component value is zero.  Consequently, we still believe the protocol to optimal for $\ket{W_4}$.

\section{Application to the transformation \texorpdfstring{$\ket{\varphi}_{1...N}\to\ket{W_N}$}{varphi}}
\label{Sect:standardW}

For a generic W-class state $\ket{\varphi}_{1...N}$, there has been promising progress on the SLOCC transformation of $\ket{\varphi}_{1...N}\to\ket{W_N}$ since the discovery 
of the unique form possessed by multipartite W-class states \cite{Kintas-2010a}. However,
the upper bound on the transformation success probability determined by the K-T monotones is not tight when the $x_0$ component of the initial state is not zero.  A canonical example of this is the transformation
of W-class state $\vec{x}=(tx_1, tx_2, \cdots, tx_n)$ into $\vec{y}=(x_1, x_2, ..., x_n)$ for $0<t<1$, which cannot be accomplished with probability $t$, and thus does not saturate the K-T monotones \cite{Kintas-2010a}.  

In the following, we improve on the general upper bound of $Nt$ set by the K-T monotones for the transformation $\ket{\vec{t}}\to\ket{W_N}$, where $\vec{t}=(t,t,\cdots, t)$.  We do this by first considering the random distillation of $\vec{t}$ into EPR pairs between party 1 and any other party.
\begin{lemma}
\label{lemma: x0nonzero}
 The optimal LOCC success probability for randomly distilling the $N$-partite W-class state $\vec{t}=(t,t, \cdots, t), 0\leq t\leq \frac{1}{N}$ into EPR pairs between party 1 and any other party is upper bounded by 
\begin{equation}
 p\leq 1-\sqrt{1-4(N-1)t^2}.
\end{equation}
\end{lemma}
\begin{proof} The proof is straightforward. We can ``merge'' together all parties other then party 1 so that we have a state unitarily equivalent to $\ket{\psi}=\sqrt{1-Nt}\ket{00}+\sqrt{t}\ket{10}+\sqrt{(N-1)t}\ket{01}$, whose smallest Schmidt component is 
$\sqrt{\tfrac{1-\sqrt{1-4(N-1)t^2}}{2}}$.  Therefore, an upper bound on the probability for distilling EPRs across the bipartition $1:23\cdots N$ is $1-\sqrt{1-4(N-1)t^2}$. 
\end{proof}

The following theorem then shows the desired result.

\begin{theorem}
 The optimal LOCC transformation probability from $N$-partite W class state $\vec{t}=(t,t, \cdots, t), 0\leq t\leq \frac{1}{N}$ into the standard W state 
$\ket{W_N}=(\frac{1}{N}, \frac{1}{N}, \cdots, \frac{1}{N})$ is upper bounded by 
\begin{equation}
 P(\vec{t}\rightarrow W_N)\leq \tfrac{N}{2}(1-\sqrt{1-4(N-1)t^2})< Nt.
\end{equation}
\end{theorem}

\begin{proof} We know that the optimal random-pair transformation of $\ket{W_N}$ into an EPR state shared between party 1 and some other party has probability $\frac{2}{N}$. If the transformation 
probability from $\vec{t}$ into $\ket{W_N}$ is higher than $\tfrac{N}{2}(1-\sqrt{1-4(N-1)t^2})$, then we can firstly transform $\vec{t}$ into $\ket{W_N}$,
and then distill EPR pairs between party 1 and the other parties with an overall successful probability larger than 
$1-\sqrt{1-4(N-1)t^2}$, contradicting with lemma 2.  Then to finish proving the theorem, we must show that $\tfrac{N}{2}(1-\sqrt{1-4(N-1)t^2})< Nt$ when $0<t<\frac{1}{N}$.  It is an elementary optimization exercise to see that $ 1- 2t - \sqrt{1-4(N-1)t^2}<0$ whenever $0<t<\frac{1}{N}$.
\end{proof}
This ``grouping'' argument given for state $\vec{t}$ can be generalized to any $\ket{\varphi}_{1...N}$ having $x_0\not=0$ in order to obtain an upper bound on the transformation probability of $\ket{\varphi}_{1...N}\to\ket{W_N}$.  While our upper bound is an improvement over the K-T monotones, it is not tight in general.  Proving optimal transformation probability when $x_0\not=0$ remains an open problem.

\section{Conclusion}

To conclude this article, let us first summarize the overall idea of the ``Least Party Out'' Protocol.  Given a generic W-class state, we first remove the $x_0$ component with some probability.  We then proceed to symmetrize by converting to standard W states $\ket{W_{|S|}^{(S)}}$.  This is what the ``Equal or Vanish'' subroutine accomplishes, and it does so in such a way that the symmetry exists only between parties connected by $\mathcal{G}$ or any subgraph of $\mathcal{G}$.  Finally, given a standard W state, the desired EPR pairings are obtained by removing parties from the entanglement in order of their connectivity in $\mathcal{G}$, the ``least'' parties being removed first.

For three qubit random distillations, our protocol is optimal, and for four qubits, it is proven optimal when $x_0=0$ for all but Configuration VI, although it still may be optimal for the standard W state.  In proving optimality for Configurations III and IV, the strategy was to compute the general expression for the LPO probability when $x_0=0$, and then show that this expression is an entanglement monotone.  We have applied the same approach to study random distillations in systems with a greater number of parties.  Unfortunately, the general expression for the LPO probability becomes quite complicated.  This can be explicitly seen from Eq. \eqref{Eq:LPOprob} in which the number of terms in the sum scales as $O(2^N)$ for a general configuration graph $\mathcal{G}$.

\subsection*{Open Questions and Concluding Remarks}

\medskip

\noindent\textbf{I.}

An obvious unresolved problem is to complete the four qubit picture by solving the random distillation of Configuration VI.  We know the LPO Protocol is not optimal for non-standard W-class states, but it is not clear why this is the case.  One possibility is the existence of $k$-cliques (a set of $k$ nodes all connected to one another) and the fact that all but one party belongs to a 3-clique.  While Configurations IV and V also have 3-cliques, each party belongs to at least one.  This may be the reason why the protocol behaves optimally in these two cases.  Understanding precisely the limitations of our protocol for Configuration VI may also prove helpful when considering the same configuration of random distillations for more general states beyond the W Class.

\medskip

\noindent\textbf{II.}

Another open problem is to generalize some of our results to a larger number of parties, especially the random distillations whose configuration graphs have relatively few edges.  For example, consider the first graph in Configurations III for which we know the LPO Protocol reduces to the Fortescue-Lo Protocol, and it is optimal.  For a six qubit system, this configuration generalizes to three disjoint pairings of the parties: $(1,2)$, $(3,4)$, and $(5,6)$.  If, for this configuration, we perform the LPO Protocol on the state $(x_1,x_2,x_3,x_4,x_5,x_6)$ with $x_1\geq x_2\geq...\geq x_6$, the resultant probability function is 
\begin{align}
 \tau_6&= 2(x_2+x_4+x_6-\frac{x_2x_4}{x_1}-\frac{x_2x_6}{x_1}-\frac{x_4x_6}{x_3})\notag\\ &+2\frac{x_2x_4x_6}{x_1x_3}+2\frac{x_4x_5x_6}{x_3^2}+\frac{2x_2x_3x_4}{3x_1^2}+\frac{2x_2x_5x_6}{3x_1^2}\notag\\
&-2\frac{x_2x_4x_5x_6}{x_1x_3^2}-\frac{14x_2x_3x_4x_5x_6}{15x_1^4}.
\end{align}
We strongly suspect that this probability is optimal, but we have no proof at this point.  As in the four qubit case, the LPO Protocol reduces to the Generalized Fortescue-Lo Protocol.  Note that for the state $\ket{W_6}$, the success probability is $2/5$.

\begin{figure}[t]
\includegraphics[scale=0.6]{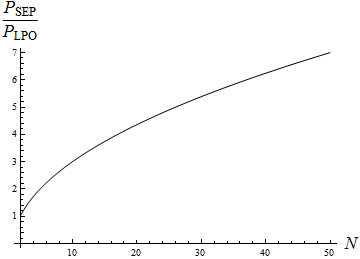}
\caption{\label{Fig:SEPvsLOCC3}
The relative difference between the optimal separable operation and the LPO Protocol.  The configuration graph consists of $N$ disjoint pairs.  Separable operations performs as $P_{SEP}=\sqrt{\frac{1}{N}}$ where as the LPO Protocol obtains the rate of $P_{LPO}=\frac{2}{2N-1}$.  We conjecture that the LPO protocol is LOCC optimal for this configuration graph, as its known to be when $N=4$.} 
\end{figure}

The generalization of this configuration to $2N$ qubits consists of a graph $\mathcal{G}_{2N}$ with $N$ disjoint pairings.  Intuitively, the LPO protocol will again reduce to the Generalized Fortescue-Lo Protocol since there exists no particular ``least'' party.  That is, the procedure will be for each party to perform weak measurements to randomly obtain three qubit W states $\ket{W_3}$ from which an EPR state can be obtained by a specified pair with probability $2/3$.  One particular trio will obtain a W-state with probability $1/\tbinom{2N}{3}$, and there are a total of $2N-2$ trios in which this particular duo can belong.  And finally, there are $N$ possible pairs.  Thus, the total probability of some specified pair $(i,j)$ obtaining an EPR state is:
\[P_{LPO}(W_{2N},\mathcal{G}_{2N})=\frac{1}{\tbinom{2N}{3}}\cdot(2N-2)\cdot\frac{2}{3}\cdot N=\frac{2}{2N-1}.\]
What is particularly interesting is when this transformation is compared to the optimal distillation probability by separable operations.  As shown in Ref. \cite{Chitambar-2011b}, this probability is given by
\[P_{SEP}(W_{2N},\mathcal{G}_{2N})=\sqrt{\frac{1}{N}}.\]
Thus, if the LPO procedure is optimal for this particular transformation, which we strongly believe it is, then we see that the performance gap between separable operations and LOCC grows arbitrarily large.  We depict this relative difference in Fig. \ref{Fig:SEPvsLOCC3}.  For example, in the four qubit case where the LPO procedure is optimal, we have $P_{LOCC}=\frac{2}{3}<P_{SEP}=\sqrt{\frac{1}{2}}$.

\medskip

\noindent\textbf{III.}

Beyond the W-class of states, very little is known about single copy random distillations.  Partial extensions of the Fortescue-Lo Protocol to symmetric Dicke states have been made, it has been shown that even within the three qubit GHZ class, distilling to randomly chosen pairs outperforms distilling to a specified pair \cite{Fortescue-2008a}.  Nevertheless, how the topology of the outcome configuration graph $\mathcal{G}$ affects these transformations has yet to be studied in general.  We hope the results of this article shed light on this question and provide a new insight into the structure of multipartite entanglement.

\begin{acknowledgments}
We thank Jonathan Oppenheim, Ben Fortescue and Sandu Popescu for providing helpful discussions during the development of this work.  We also thank the financial support from funding agencies including NSERC, QuantumWorks, the CRC program and CIFAR. 
\end{acknowledgments}

\newpage

\appendix

\section{Proof of Theorem \ref{Thm:tau}}
\label{Apx:tau}
\begin{widetext}
We consider case-by-case measurements in which each party acts according to \eqref{Eq:constraints}.  The function $\tau$ transforms as $\tau\to\tau_\lambda$ for $\lambda\in\{1,2\}$, and we are interested in the average change: $\overline{\tau_\lambda}=p_1\tau_1+p_2\tau_2$.  By the universality of weak measurements \cite{Bennett-1999a, Oreshkov-2005a}, it is sufficient to prove $\tau$ monotonic in the weak measurement setting, i.e. with $(a_1,c_1,a_2,c_2)$ in some neighborhood of $(1/2,1/2,1/2,1/2)$.  We consider three cases. 

\noindent\textbf{Case I, $x_{n_1}>x_{n_1'}$:}
First consider when party $n_1$ performs a measurement.  We can assume the measurement is weak enough such that $n_1$, $n_1'$, $p$ and $p'$ are the same for both pre-measurement and post-measurement states.  Consider the measurement outcome $\lambda\in\{0,1\}$ with $a_\lambda > c_\lambda$.  Then
\[p_\lambda\tau_\lambda=2a_\lambda x_{p}+2a_\lambda x_{p'}-2\frac{a_\lambda^2}{c_\lambda}\left(1-\frac{1}{3}\frac{a_\lambda}{c_\lambda}\frac{x_{n_1'}}{x_{n_1}}\right)\frac{x_{p}x_{p'}}{x_{n_1}}.\]
We have $\frac{\partial\overline{\tau_\lambda}}{\partial c_\lambda}|_{a_\lambda=1/2,c_\lambda=1/2}\geq 0$ which implies that $\tau-\overline{\tau_\lambda}$ will be minimized by the choice $c_1+c_2=1$.  Then writing $c\equiv c_1<a\equiv a_1$, we have that
\begin{align}\tau-\overline{\tau_\lambda}=2\left(-1+\frac{a^2}{c}+\frac{(1-a)^2}{(1-c)}\right)\frac{x_{p}x_{p'}}{x_{n_1}}+\frac{2}{3}\left(1-\frac{a^3}{c^2}-\frac{(1-a)^3}{(1-c)^2}\right)\frac{x_{p}x_{p'}x_{n_1'}}{x_{n_1}^2}.
\end{align}
Expanding this expression about the point $(1/2,1/2)$ to second order gives
\begin{align}\tau-\overline{\tau_\lambda}\approx&\;\; 8\left((a-\tfrac{1}{2})^2+(c-\tfrac{1}{2})^2-2(a-\tfrac{1}{2})(c-\tfrac{1}{2})\right)\left(1-\frac{x_{n_1'}}{x_{n_1}}\right)\frac{x_{p}x_{p'}}{x_{n_1}}\notag\\
=&\;\;8\left(a-c\right)^2\left(1-\frac{x_{n_1'}}{x_{n_1}}\right)\frac{x_{p}x_{p'}}{x_{n_1}}\geq 0.
\end{align}
Now consider when the other parties measure.  Since the coefficient of $x_p$ is non-negative, the monotonicity of $\tau$ when party $p$ measures follows from the K-T monotones.  For $n_1'$ and $p'$, there are two possibilities.  \textbf{Subcase, $x_{n_1'}>x_{p'}$:}  Here we can assume the measurements are weak enough such that their ordering does not change.  Then since the coefficients of $x_{p'}$ and $x_{n_1'}$ are non-negative, the K-T monotones imply the monotonicity of $\tau$.  \textbf{Subcase, $x_{n_1'}=x_{p'}$:} We have $\tau=2x_p+2x_{p'}-2\frac{x_px_{p'}}{x_{n_1}}+\frac{2}{3}\frac{x_px_{p'}^2}{x_{n_1}^2}$.  It is easy to see that again $\tau-\overline{\tau_\lambda}$ is minimized when $c_1+c_2=1$.  So parameterizing the measurement by $a$ and $c$ with $a>c$, we have that the average change in $x_{n_1'}$ is $(a+1-c)x_{p'}$ while the average change in $x_{p'}$ is $(1-a+c)x_{p'}$.  It follows that 
\[\tau-\overline{\tau_\lambda}=2x_{p'}(a-c)-2(a-c)\frac{x_{p}x_{p'}}{x_{n_1}}\geq 0.\]

\noindent\textbf{Case II, $x_{n_1}=x_p$:}
Here, $\tau=2x_p+\frac{2}{3}\frac{x_{p'}x_{n_1'}}{x_{n_1}}$.  When either party $n_1$ or $p$ measures with $a_\lambda>c_\lambda$, the new components are $x_{\lambda,n_1}=\frac{a_\lambda}{p_\lambda} x_p$, $x_{\lambda,p}=\frac{c_\lambda}{p_\lambda}x_p$, $x_{\lambda,n_1'}=\frac{a_\lambda}{p_\lambda}x_{p'}$ and $x_{\lambda,p'}=\frac{a_\lambda}{p_\lambda}x_{n_1'}$.  Thus, 
\[p_\lambda\tau_\lambda=2(a_\lambda-c_\lambda)x_{n_1'}+2c_\lambda x_p+\frac{2c_\lambda}{3}\frac{x_{p'}x_{n_1'}}{x_p}.\]
Since $x_p\geq x_{n_1'}$, we have $\frac{\partial\overline{\tau_\lambda}}{\partial c_\lambda}|_{a_\lambda=1/2,c_\lambda=1/2}\geq 0$.  Again, this means that $\tau-\overline{\tau_\lambda}$ will be minimized by the choice $c_1+c_2=1$. Taking $a>c$, we have that
\begin{align}\tau-\overline{\tau_\lambda}&=2(1-c-(1-a))x_{p}-2\left(1-c-\frac{(1-a)^2}{1-c}\right)x_{n_1'}+\frac{2}{3}\left(1-c-\frac{(1-a)^3}{(1-c)^2}\right)\frac{x_{p}x_{p'}}{x_{n_1'}}
\end{align}
which to first order about the point $(1/2,1/2)$ takes the form \begin{equation}
2(a-c)(x_{p}-2x_{n_1'}+\frac{x_{p}x_{p'}}{x_{n_1'}}\geq 2(a-c)\frac{(x_{n_p}-x_{n_1'})^2}{x_{n_p}}\geq 0.
\end{equation}
If either $x_{p'}$ or $x_{n_1'}$ measures, then the monotonicity of $\tau$ follows from the K-T monotones.

\noindent\textbf{Case III, $x_{n_1}=x_{n_1'}$:}  We have $\tau=2x_{p}+2x_{p'}-\frac{4}{3}\frac{x_{p}x_{p'}}{x_{n_1}}$.  When either party $n_1$ or $n_1'$ measures, parties $p$ and $p'$ remain the same.  With $a_\lambda>c_\lambda$, we have
\[p_\lambda\tau=2a_\lambda x_p+2a_\lambda x_{p'}-2a_\lambda\frac{x_px_{p'}}{x_{n_1}}+\frac{2}{3}c_\lambda\frac{x_px_{p'}}{x_{n_1}}\]
which has $\frac{\partial\overline{\tau_\lambda}}{\partial c_\lambda}|_{a_\lambda=1/2,c_\lambda=1/2}\geq 0$.  So again we assume $c_1=1-c_2\equiv c<a$ and we find that 
\begin{align}
\tau-\overline{\tau_\lambda}&=-2\left(\frac{2}{3}-a-\frac{(1-a)^2}{1-c}\right)\frac{x_{p}x_{p'}}{x_{n_1}}-\frac{2}{3}\left(c+\frac{(1-a)^3}{(1-c)^2}\right)\frac{x_{p}x_{p'}}{x_{n_1}}
\end{align}
which to third order about the point $(1/2,1/2)$ takes the form 
\begin{equation}
\frac{8}{3}\left((a-\tfrac{1}{2})^3-(c-\tfrac{1}{2})^3\right)-8\left((c-\tfrac{1}{2})(a-\tfrac{1}{2})^2-(a-\tfrac{1}{2})(c-\tfrac{1}{2})^2\right)\frac{x_{p}x_{p'}}{x_{n_1}}=\frac{8}{3}(a-c)^3\frac{x_{n_p}x_{n_p'}}{x_{n_1}}\geq 0.
\end{equation}
Finally, since the coefficients of $x_p$ and $x_{p'}$ are positive in $\tau$, by the K-T monotones, $\tau$ is monotonic when either of these parties measures.
\end{widetext}

\bibliography{EricQuantumBib}

\begin{thebibliography}{18}
\expandafter\ifx\csname natexlab\endcsname\relax\def\natexlab#1{#1}\fi
\expandafter\ifx\csname bibnamefont\endcsname\relax
  \def\bibnamefont#1{#1}\fi
\expandafter\ifx\csname bibfnamefont\endcsname\relax
  \def\bibfnamefont#1{#1}\fi
\expandafter\ifx\csname citenamefont\endcsname\relax
  \def\citenamefont#1{#1}\fi
\expandafter\ifx\csname url\endcsname\relax
  \def\url#1{\texttt{#1}}\fi
\expandafter\ifx\csname urlprefix\endcsname\relax\def\urlprefix{URL }\fi
\providecommand{\bibinfo}[2]{#2}
\providecommand{\eprint}[2][]{\url{#2}}

\bibitem[{\citenamefont{Bennett et~al.}(1993)\citenamefont{Bennett, Brassard,
  Cr\'epeau, Jozsa, Peres, and Wootters}}]{Bennett-1993a}
\bibinfo{author}{\bibfnamefont{C.~H.} \bibnamefont{Bennett}},
  \bibinfo{author}{\bibfnamefont{G.}~\bibnamefont{Brassard}},
  \bibinfo{author}{\bibfnamefont{C.}~\bibnamefont{Cr\'epeau}},
  \bibinfo{author}{\bibfnamefont{R.}~\bibnamefont{Jozsa}},
  \bibinfo{author}{\bibfnamefont{A.}~\bibnamefont{Peres}}, \bibnamefont{and}
  \bibinfo{author}{\bibfnamefont{W.~K.} \bibnamefont{Wootters}},
  \bibinfo{journal}{Phys. Rev. Lett.} \textbf{\bibinfo{volume}{70}},
  \bibinfo{pages}{1895} (\bibinfo{year}{1993}).

\bibitem[{\citenamefont{Bennett and Wiesner}(1992)}]{Bennett-1992b}
\bibinfo{author}{\bibfnamefont{C.~H.} \bibnamefont{Bennett}} \bibnamefont{and}
  \bibinfo{author}{\bibfnamefont{S.~J.} \bibnamefont{Wiesner}},
  \bibinfo{journal}{Phys. Rev. Lett.} \textbf{\bibinfo{volume}{69}},
  \bibinfo{pages}{2881} (\bibinfo{year}{1992}).

\bibitem[{\citenamefont{Gour et~al.}(2005)\citenamefont{Gour, Meyer, and
  Sanders}}]{Gour-2005a}
\bibinfo{author}{\bibfnamefont{G.}~\bibnamefont{Gour}},
  \bibinfo{author}{\bibfnamefont{D.~A.} \bibnamefont{Meyer}}, \bibnamefont{and}
  \bibinfo{author}{\bibfnamefont{B.~C.} \bibnamefont{Sanders}},
  \bibinfo{journal}{Phys. Rev. A} \textbf{\bibinfo{volume}{72}},
  \bibinfo{pages}{042329} (\bibinfo{year}{2005}).

\bibitem[{\citenamefont{Chitambar et~al.}(2010)\citenamefont{Chitambar, Duan,
  and Shi}}]{Chitambar-2010a}
\bibinfo{author}{\bibfnamefont{E.}~\bibnamefont{Chitambar}},
  \bibinfo{author}{\bibfnamefont{R.}~\bibnamefont{Duan}}, \bibnamefont{and}
  \bibinfo{author}{\bibfnamefont{Y.}~\bibnamefont{Shi}},
  \bibinfo{journal}{Phys. Rev. A} \textbf{\bibinfo{volume}{81}},
  \bibinfo{pages}{052310} (\bibinfo{year}{2010}).

\bibitem[{\citenamefont{Smolin et~al.}(2005)\citenamefont{Smolin, Verstraete,
  and Winter}}]{Smolin-2005a}
\bibinfo{author}{\bibfnamefont{J.~A.} \bibnamefont{Smolin}},
  \bibinfo{author}{\bibfnamefont{F.}~\bibnamefont{Verstraete}},
  \bibnamefont{and} \bibinfo{author}{\bibfnamefont{A.}~\bibnamefont{Winter}},
  \bibinfo{journal}{Phys. Rev. A} \textbf{\bibinfo{volume}{72}},
  \bibinfo{pages}{052317} (\bibinfo{year}{2005}).

\bibitem[{\citenamefont{Yang and Eisert}(2009)}]{Yang-2009a}
\bibinfo{author}{\bibfnamefont{D.}~\bibnamefont{Yang}} \bibnamefont{and}
  \bibinfo{author}{\bibfnamefont{J.}~\bibnamefont{Eisert}},
  \bibinfo{journal}{Phys. Rev. Lett.} \textbf{\bibinfo{volume}{103}},
  \bibinfo{pages}{220501} (\bibinfo{year}{2009}).

\bibitem[{\citenamefont{Fortescue and Lo}(2007)}]{Fortescue-2007a}
\bibinfo{author}{\bibfnamefont{B.}~\bibnamefont{Fortescue}} \bibnamefont{and}
  \bibinfo{author}{\bibfnamefont{H.-K.} \bibnamefont{Lo}},
  \bibinfo{journal}{Phys. Rev. Lett.} \textbf{\bibinfo{volume}{98}},
  \bibinfo{pages}{260501} (\bibinfo{year}{2007}).

\bibitem[{\citenamefont{Fortescue and Lo}(2008)}]{Fortescue-2008a}
\bibinfo{author}{\bibfnamefont{B.}~\bibnamefont{Fortescue}} \bibnamefont{and}
  \bibinfo{author}{\bibfnamefont{H.-K.} \bibnamefont{Lo}},
  \bibinfo{journal}{Phys. Rev. A} \textbf{\bibinfo{volume}{78}},
  \bibinfo{pages}{012348} (\bibinfo{year}{2008}).

\bibitem[{\citenamefont{Fortescue}(2009)}]{Fortescue-2009a}
\bibinfo{author}{\bibfnamefont{B.}~\bibnamefont{Fortescue}}, Ph.D. thesis,
  \bibinfo{school}{The University of Toronto} (\bibinfo{year}{2009}).

\bibitem[{\citenamefont{Oppenheim}(2007)}]{Oppenheim-2011a}
\bibinfo{author}{\bibfnamefont{J.}~\bibnamefont{Oppenheim}}
  (\bibinfo{year}{2007}), \bibinfo{note}{private correspondence}.

\bibitem[{\citenamefont{Wieczorek et~al.}(2009)\citenamefont{Wieczorek,
  Krischek, Kiesel, Michelberger, T\'oth, and Weinfurter}}]{Wieczorek-2009a}
\bibinfo{author}{\bibfnamefont{W.}~\bibnamefont{Wieczorek}},
  \bibinfo{author}{\bibfnamefont{R.}~\bibnamefont{Krischek}},
  \bibinfo{author}{\bibfnamefont{N.}~\bibnamefont{Kiesel}},
  \bibinfo{author}{\bibfnamefont{P.}~\bibnamefont{Michelberger}},
  \bibinfo{author}{\bibfnamefont{G.}~\bibnamefont{T\'oth}}, \bibnamefont{and}
  \bibinfo{author}{\bibfnamefont{H.}~\bibnamefont{Weinfurter}},
  \bibinfo{journal}{Phys. Rev. Lett.} \textbf{\bibinfo{volume}{103}},
  \bibinfo{pages}{020504} (\bibinfo{year}{2009}).

\bibitem[{\citenamefont{Bastin et~al.}(2009)\citenamefont{Bastin, Thiel, von
  Zanthier, Lamata, Solano, and Agarwal}}]{Bastin-2009a}
\bibinfo{author}{\bibfnamefont{T.}~\bibnamefont{Bastin}},
  \bibinfo{author}{\bibfnamefont{C.}~\bibnamefont{Thiel}},
  \bibinfo{author}{\bibfnamefont{J.}~\bibnamefont{von Zanthier}},
  \bibinfo{author}{\bibfnamefont{L.}~\bibnamefont{Lamata}},
  \bibinfo{author}{\bibfnamefont{E.}~\bibnamefont{Solano}}, \bibnamefont{and}
  \bibinfo{author}{\bibfnamefont{G.~S.} \bibnamefont{Agarwal}},
  \bibinfo{journal}{Phys. Rev. Lett.} \textbf{\bibinfo{volume}{102}},
  \bibinfo{pages}{053601} (\bibinfo{year}{2009}).

\bibitem[{\citenamefont{Kinta\c{s} and Turgut}(2010)}]{Kintas-2010a}
\bibinfo{author}{\bibfnamefont{S.}~\bibnamefont{Kinta\c{s}}} \bibnamefont{and}
  \bibinfo{author}{\bibfnamefont{S.}~\bibnamefont{Turgut}},
  \bibinfo{journal}{J. Math. Phys.} \textbf{\bibinfo{volume}{51}},
  \bibinfo{pages}{092202} (\bibinfo{year}{2010}).

\bibitem[{\citenamefont{Chitambar et~al.}(2011)\citenamefont{Chitambar, Cui,
  and Lo}}]{Chitambar-2011b}
\bibinfo{author}{\bibfnamefont{E.}~\bibnamefont{Chitambar}},
  \bibinfo{author}{\bibfnamefont{W.}~\bibnamefont{Cui}}, \bibnamefont{and}
  \bibinfo{author}{\bibfnamefont{H.-K.} \bibnamefont{Lo}}
  (\bibinfo{year}{2011}), \bibinfo{note}{arxiv.org/quant-ph}.

\bibitem[{\citenamefont{Anderson and Oi}(2008)}]{Anderson-2008a}
\bibinfo{author}{\bibfnamefont{E.}~\bibnamefont{Anderson}} \bibnamefont{and}
  \bibinfo{author}{\bibfnamefont{D.~K.~L.} \bibnamefont{Oi}},
  \bibinfo{journal}{Phys. Rev. A} \textbf{\bibinfo{volume}{77}},
  \bibinfo{pages}{052104} (\bibinfo{year}{2008}).

\bibitem[{\citenamefont{Cui et~al.}(2010)\citenamefont{Cui, Helwig, and
  Lo}}]{Cui-2010a}
\bibinfo{author}{\bibfnamefont{W.}~\bibnamefont{Cui}},
  \bibinfo{author}{\bibfnamefont{W.}~\bibnamefont{Helwig}}, \bibnamefont{and}
  \bibinfo{author}{\bibfnamefont{H.-K.} \bibnamefont{Lo}},
  \bibinfo{journal}{Phys. Rev. A} \textbf{\bibinfo{volume}{81}},
  \bibinfo{pages}{012111} (\bibinfo{year}{2010}).

\bibitem[{\citenamefont{Bennett et~al.}(1999)\citenamefont{Bennett, DiVincenzo,
  Fuchs, Mor, Rains, Shor, Smolin, and Wootters}}]{Bennett-1999a}
\bibinfo{author}{\bibfnamefont{C.~H.} \bibnamefont{Bennett}},
  \bibinfo{author}{\bibfnamefont{D.~P.} \bibnamefont{DiVincenzo}},
  \bibinfo{author}{\bibfnamefont{C.~A.} \bibnamefont{Fuchs}},
  \bibinfo{author}{\bibfnamefont{T.}~\bibnamefont{Mor}},
  \bibinfo{author}{\bibfnamefont{E.}~\bibnamefont{Rains}},
  \bibinfo{author}{\bibfnamefont{P.~W.} \bibnamefont{Shor}},
  \bibinfo{author}{\bibfnamefont{J.~A.} \bibnamefont{Smolin}},
  \bibnamefont{and} \bibinfo{author}{\bibfnamefont{W.~K.}
  \bibnamefont{Wootters}}, \bibinfo{journal}{Phys. Rev. A}
  \textbf{\bibinfo{volume}{59}}, \bibinfo{pages}{1070} (\bibinfo{year}{1999}).

\bibitem[{\citenamefont{Oreshkov and Brun}(2005)}]{Oreshkov-2005a}
\bibinfo{author}{\bibfnamefont{O.}~\bibnamefont{Oreshkov}} \bibnamefont{and}
  \bibinfo{author}{\bibfnamefont{T.}~\bibnamefont{Brun}},
  \bibinfo{journal}{Phys. Rev. Lett.} \textbf{\bibinfo{volume}{95}},
  \bibinfo{pages}{110409} (\bibinfo{year}{2005}).

\end{thebibliography}

\end{document}